\documentclass[12pt]{article}
\usepackage[latin9]{inputenc}
\usepackage{amsmath}
\usepackage{amssymb}
\usepackage{esint}

\makeatletter
\usepackage{amsfonts}\usepackage{color}\usepackage{enumitem}\usepackage{calrsfs}
\usepackage{amsthm}
\allowdisplaybreaks

\definecolor{labelkey}{rgb}{0.6,0,1}

\renewcommand{\d}{{\rm d}}

\newcommand{\E}{{\mathbb E}}
\newcommand{\F}{{\mathbb F}}

\def\L{{\mathbb L}}
\renewcommand{\P}{{\mathbb P}}
\newcommand{\Q}{{\mathbb Q}}

\newcommand{\R}{{\bf R}}

\newcommand{\N}{{\mathbb N}}

\newcommand{\Acal}{{\mathcal A}}
\newcommand{\Bcal}{{\mathcal B}}

\newcommand{\Dcal}{{\mathcal D}}
\newcommand{\Ecal}{{\mathcal E}}
\newcommand{\Fcal}{{\mathcal F}}

\newcommand{\Kcal}{{\mathcal K}}

\newcommand{\Mcal}{{\mathcal M}}

\newcommand{\lc}{[\![}
\newcommand{\rc}{]\!]}
\newcommand{\1}[1]{{\boldsymbol 1_{\{#1\}}}}
\newcommand{\oo}{{\bf 1}}

\def\NA{{\rm NA}}
\def\ND{{\rm ND}}
\def\NUPBR{{\rm NUPBR}}
\def\NFLVR{{\rm NFLVR}}

\newtheorem{theorem}{Theorem}

\newtheorem{corollary}{Corollary}

\newtheorem{definition}{Definition}

\newtheorem{lemma}{Lemma}

\newtheorem{proposition}{Proposition}

\numberwithin{equation}{section}
\numberwithin{theorem}{section}
\numberwithin{proposition}{section}
\numberwithin{lemma}{section}
\numberwithin{definition}{section}

\makeatother

\begin{document}

\title{Informational Efficiency under Short Sale Constraints%
\thanks{The authors would like to thank Kostas Kardaras, Sergio Pulido and Johannes Ruf
for helpful comments and stimulating discussions. Martin Larsson gratefully
acknowledges support from the European Research Council under the
European Union's Seventh Framework Programme (FP/2007-2013) / ERC
Grant Agreement n. 307465-POLYTE.%
}}

\author{Robert A.~Jarrow%
\thanks{Johnson School of Management, Cornell University, Ithaca, NY 14853
and Kamakura Corporation, Honolulu, Hawaii 96815. E-mail: raj15@cornell.edu%
}\hspace{2cm}Martin Larsson%
\thanks{Swiss Finance Institute @ Ecole Polytechnique Fédérale de Lausanne.
E-mail: martin.larsson@epfl.ch%
}}

\date{December 31, 2013}

\maketitle
 
\begin{abstract}
A constrained informationally efficient market is defined to be one
whose price process arises as the outcome of some equilibrium where
agents face restrictions on trade. This paper investigates the case
of short sale constraints, a setting which despite its simplicity,
generates new insights. In particular, it is shown that short sale
constrained informationally efficient markets always admit equivalent
supermartingale measures and local martingale deflators, but not necessarily
local martingale measures. And if in addition some local martingale
deflator turns the price process into a true martingale, then the
market is constrained informationally efficient. Examples are given to illustrate
the subtle phenomena that can arise in the presence of short sale
constraints, with particular attention to representative agent equilibria
and the different notions of no arbitrage.


JEL CLASSICATION: G11, G12, G14, D52, D53 
\end{abstract}

\section{Introduction}

Since Fama \cite{Fama:1970} published his seminal paper on market
efficiency four decades ago, a vast body of empirical research has
emerged evaluating the informational efficiency of real-world markets
with mixed results. These mixed results are partly due to the the
notorious {\em joint model hypothesis} problem---to test for informational
efficiency one must jointly hypothesize an equilibrium model. Competing
theories to the rational competitive markets paradigm have also been
proposed to reconcile this conflicting evidence such as the behavioral
approach to finance and economics (excellent surveys include Hirshleifer
\cite{hirshleifer2001investor}, Shleifer \cite{shleifer2003inefficient},
Barberis and Thaler \cite{barberis2003survey}) and the adaptive market
hypothesis (see, for instance, Lo \cite{Lo2004}). In all cases, informational
efficiency, i.e.~the degree to which information about fundamentals
is quickly and accurately reflected in asset prices, has proved to
be an immensely powerful tool for understanding the behavior of asset
markets.

To circumvent this joint model hypothesis, Jarrow and Larsson \cite{Jarrow:2012fk}
formalized the definition of informational efficiency, consistent
with Fama's original ideas, as the existence of {\em some} equilibrium
supporting the given asset price process. In a continuous-time setting,
the notion of no arbitrage unfortunately becomes very delicate (see
Harrison and Kreps\cite{HarrisonKreps1979}, Delbean and Schachermayer
\cite{Delbaen/Schachermayer:1994,Delbaen/Schachermayer:1998}) so
one might expect the same to be true for the existence of a supporting
equilibrium. However, following earlier work in this direction (Harrison
and Kreps~\cite{HarrisonKreps1979}, Kreps \cite{Kreps1981}, Ross
\cite{ross2009neoclassical}) in \cite{Jarrow:2012fk} a surprisingly
simple answer is given: an equilibrium supporting a given price process
exists if and only if the price process satisfies the No Free Lunch
with Vanishing Risk (NFLVR) property of Delbaen and Schachermayer,
as well as Merton's No Dominance (ND) condition~\cite{Merton:1973},
see also Jarrow, Protter, Shimbo \cite{Jarrow/Protter/Shimbo:2006,Jarrow/Protter/Shimbo:2010}.
A key implication of this characterization is that informational efficiency
can be rejected without committing to any particular equilibrium model
by merely proving the existence of arbitrage opportunities and/or
dominated strategies. At least for rejecting informational efficiency,
this circumvents the joint model hypothesis problem.

For accepting informational efficiency, however, exhausting all possible
arbitrage opportunities and/or dominated strategies is not practical.
A second characterization of informational efficiency also proven
in \cite{Jarrow:2012fk} can be used in this regard. It is shown that
a supporting equilibrium exists if and only if an equivalent probability
measure exists such that the price process, normalized by the money
market account's value, is a martingale. This is called an \emph{equivalent
martingale measure}. With this characterization, informational efficiency
can be accepted using a \emph{conditional joint model hypothesis}.
Indeed, one first assumes a particular stochastic process for the
asset's price evolution. Then, this evolution is empirically validated/rejected.
If this evolution is validated, then the stochastic processes' parameters
are checked for consistency with the existence of an equivalent martingale
measure. If consistent, then informational efficiency is accepted.
In contrast to the joint model hypothesis, this approach has the advantage
that the assumed evolution can be validated/rejected independently
of informational efficiency.

A crucial assumption in the above formulation is that markets are
frictionless, i.e. there are no transaction costs and trading is unrestricted.
When this assumption provides a reasonable approximation, the above
methodology is appropriate for tesing informational efficiency. Yet,
not for all markets and not for a given market at all times is trading
unrestricted. Indeed, in most emerging markets short selling is impossible
or not allowed (see Bris, Goetzmann, and Zhu \cite{BrisGoetzmannZhu2007}),
and during the recent financial crisis regulators in well developed
markets prohibited short sales to halt declining prices (see Beber
and Pagano \cite{BeberPagano2011}, Boehmer, Jones and Zhang \cite{boehmer2013}).
When trading is constrained, using the above methodology to test for
informational efficiency may lead to false rejections. The rejection
being due to the existence of trading constraints and not market inefficiency.

The purpose of the present paper is to extend the model in \cite{Jarrow:2012fk}
to include trading constraints, more specifically short sale constraints.
We now define a {\em constrained informationally efficient} market
as one where there exists an equilibrium with agents facing short
sale constraints that supports the given price process. This seemingly
modest alteration brings new insights and subtleties in the characterization
of an informationally efficient market.

Our first main result, Theorem~\ref{T:NUPBR}, shows that constrained
informational efficiency implies the existence of a local martingale
deflator, turning asset prices into local martingales. This may come
as a surprise since, due to the short sale constraint, only supermartingale
deflators should be guaranteed. The key insight here is that market
clearing always results in some agent being locally unconstrained,
and this suffices to yield a local martingale deflator. It is well-known
that the existence of a local martingale deflator does not rule out
unconstrained arbitrage. Nonetheless, the above intuition regarding
the role of market clearing suggests that unconstrained arbitrage
strategies are also impossible in equilibrium---indeed, if such a
strategy were to exist, an investor with strictly positive holdings
could add ``a little bit'' of it, and thus do better without violating
the short sale constraint. However, this reasoning fails. In Section~\ref{S:ex2}
we give an example of an equilibrium with a single representative
agent with logarithmic utility, optimally holding the full supply
of the risky asset, but where arbitrage using unconstrained strategies
is nonetheless possible (i.e., the condition $\NA$ fails.)

Less surprisingly, we find that constrained informational efficiency
always implies a constrained version of $\NFLVR$. However, the situation
regarding $\ND$ is less clear. We show that a constrained version
of $\ND$ holds whenever a representative agent equilibrium exists,
but beyond that not much can be said. In fact, our strongest results
are obtained in the setting of constrained representative agent equilibria.
This situation is addressed by our second main result, Theorem~\ref{T:reprchar},
which shows that a given price process is supported by some such equilibrium
if and only if, in addition to the above conditions, there is a local
martingale deflator that turns the price process into a true martingale.
Moreover, in a multi-agent complete market equilibrium, these conditions
will be satisfied whenever all agents have strictly positive holdings
in the risky asset, see Proposition~\ref{P:complete}. Finally, if
a given equilibrium yields a price process that could, theoretically,
be supported by some representative agent equilibrium, it is of interest
to know more about this representative agent. In Theorem~\ref{T:aggr}
we show that the representative agent can always be constructed by
aggregating the individual agent utilities using the well-known procedure
pioneered in Negishi \cite{Negishi1960}.

With respect to a set of necessary and then sufficient conditions
useful for testing constrained informational efficiency, our insights
can be summarized as follows. A constrained informationally efficient
market implies both No Unbounded Profit with Bounded Risk ($\NUPBR$)
and constrained $\NFLVR$. If one can discover trading strategies
violating either condition, then constrained informational efficiency
can be rejected without the joint model hypothesis problem. In the
converse direction, if there exists an equivalent supermartingale
measure and a local martingale deflator that turns the price process
into a martingale, then the market is constrained informationally
efficient. In principle, this sufficient condition can be used, in
conjunction with a validated price process, to accept an informationally
efficient market. Testing for constrained informational efficiency
using these new insights awaits subsequent research.

An outline for this paper is as follows. Section~\ref{S:model} describes
the pricing model and reviews the relevant no-arbitrage type conditions
and their characterizations. Section~\ref{S:eff} introduces the
agents and the notion of equilibrium and constrained informational
efficiency. Necessary conditions for efficiency are discussed. Section~\ref{S:repr}
deals with representative agent equilibria, and a characterization
theorem is given. Aggregation of individual utilities is considered.
Section~\ref{S:ex} contains examples illustrating to what extent
our results are sharp. Section~\ref{S:concl} concludes and lists
some open questions suggested by our analysis.

\section{The Model}

\label{S:model}

We consider a continuous time and continuous trading economy on a
finite horizon. There are a finite number of agents in the economy,
trading in a competitive market with no transaction costs. The market
is constrained in that there are no short sales allowed for risky
assets. Otherwise, it is assumed to be a frictionless market. We are
given a complete filtered probability space $(\Omega,\Fcal,\F,\P)$,
where the filtration $\F=(\Fcal_{t})_{0\le t\le T}$, defined on a
bounded time interval $[0,T]$, satisfies the usual hypotheses. Here
$\P$ is the statistical probability measure. The traders in the economy
have the information set $\F$ and beliefs $\P_{k}$ assumed to be
equivalent to $\P$. It is this information set that is relevant for
the subsequent definition of constrained informational efficiency.%
\footnote{This structure can be extended to allow differential information across
traders analogous to that discussed in \cite{Jarrow:2012fk}, p.~11; see also \cite{Frahm:2013} for related results.
For simplicity of presentation, this extension is not discussed further
herein.%
}

The financial market is assumed to consist of one risk-free asset
and one risky asset. Prices are given in units of the risk-free asset,
so that, in particular, the risk-free asset has a constant price equal
to unity. The price at time $t$~of the risky asset is denoted by~$S_{t}$,
and it is assumed that $S=(S_{t})_{0\le t\le T}$ is a strictly positive
continuous semimartingale with respect to the filtration $\F$.

A (self-financing) \emph{trading strategy} $H$ is an $S$-integrable
predictable process representing the number of shares of the risky
asset held at each point in time. It is called \emph{$a$-admissible}
if $(H\cdot S)_{t}\ge-a$ for all $t\in[0,T]$, and \emph{admissible}
if it is $a$-admissible for some $a\ge0$. It is called \emph{constrained}
if $H_{t}\ge0$ for all $t\in[0,T]$, i.e.~if it satisfies a basic
short sale restriction. We define 
\begin{align*}
\Acal & =\{\text{all admissible strategies }H\},\\
\Acal_{C} & =\{\text{all constrained admissible strategies }H\}.
\end{align*}
Note the slight abuse of terminology: in order to specify the trading
strategy one also needs the number $H_{t}^{0}$ of shares held in
the risk-free asset at each time~$t$. However, since only self-financing
strategies are considered, the wealth process $X_{t}=H_{t}^{0}+H_{t}S_{t}$
satisfies 
\[
X_{t}=x+(H\cdot S)_{t},
\]
where $x=W_{0}$ is the initial capital. This relation determines
$H^{0}$. Note that we do not impose short sale restrictions on the
risk-free asset (see Cuoco \cite{Cuoco:1997} for a setting where
such constraints are covered.)

\subsection{No Arbitrage}

Our primary interest is how short sale constraints affect no-arbitrage
type restrictions in equilibrium. We now review the various notions
of no arbitrage.

\begin{definition} \label{D:arb}
We say
that 
\begin{itemize}
\item $\NA$ ($\NA_{C}$) holds if $(H\cdot S)_{T}\ge0$ implies $(H\cdot S)_{T}=0$
for any $H\in\Acal$ ($H\in\Acal_{C}$); 
\item $\ND$ ($\ND_{C}$) holds if $(H\cdot S)_{T}\ge S_{T}-S_{0}$ implies
$(H\cdot S)_{T}=S_{T}-S_{0}$ for any $H\in\Acal$ ($H\in\Acal_{C}$); 
\item $\NUPBR$ ($\NUPBR_{C}$) holds if the set $\Kcal$ ($\Kcal_{C}$)
is bounded in $\L^{0}$, where 
\begin{align*}
\Kcal & =\{1+(H\cdot S)_{T}:H\text{ is }1\text{-admissible}\},\\
\Kcal_{C} & =\{1+(H\cdot S)_{T}:H\text{ is constrained }1\text{-admissible}\};
\end{align*}

\item $\NFLVR$ ($\NFLVR_{C}$) holds if both $\NA$ and $\NUPBR$ ($\NA_{C}$
and $\NUPBR_{C}$) hold. 
\end{itemize}
\end{definition}

Here $\NA$ stands for No~Arbitrage, $\ND$ for No Dominance, $\NUPBR$
for No Unbounded Profit with Bounded Risk, and $\NFLVR$ for No Free
Lunch with Vanishing Risk. While $\NA$ is an old concept in finance,
$\NFLVR$ was introduced in Delbaen and Schachermayer\cite{Delbaen/Schachermayer:1994}
and then in a more general setting in~\cite{Delbaen/Schachermayer:1998}.
The condition $\NUPBR$ was studied in a general setting in Karatzas
and Kardaras~\cite{Karatzas/Kardaras:2007}. The $\ND$ condition
is perhaps less well-known, but was in fact already introduced in
Merton~\cite{Merton:1973}. Its intuitive economic meaning is that
it is impossible to find a trading strategy which dominates buying
and holding the risky asset itself. We also note that by a well-known
characterization result, see for instance \cite[Proposition~4.2]{Karatzas/Kardaras:2007},
$\NFLVR$ ($\NFLVR_{C}$) as defined above is equivalent to the condition
that there is no sequence $H^{n}\in\Acal$ ($H^{n}\in\Acal_{C}$)
such that $f_{n}=(H^{n}\cdot S)_{T}$ satisfies $\lim_{n}\|\max(-f_{n},0)\|_{\L^{\infty}}=0$
and $\lim_{n}f_{n}=f$ almost surely for some $f\ge0$ with $\P(f>0)>0$.

The following notion plays in important role in the analysis of the
$\ND$ condition, among other things.

\begin{definition} A strategy $H\in\Acal$ ($H\in\Acal_{C}$)
is called \emph{maximal} (\emph{C-maximal}) if $(K\cdot S)_{T}\ge(H\cdot S)_{T}$
implies $(K\cdot S)_{T}=(H\cdot S)_{T}$ for any $K\in\Acal$ ($K\in\Acal_{C}$).
\end{definition}

We can now phrase the $\ND$ condition as saying that holding a constant
number of shares in the risky asset is a maximal strategy. Similarly,
the $\NA$ condition means that investing all one's wealth in the
risk-free asset is a maximal strategy.

\subsection{Dual Characterizations}

The various meanings of no arbitrage appearing in Definition~\ref{D:arb}
can be characterized in terms of equivalent probability measures,
or, more generally, deflator processes. In order to state these results
we introduce the following sets. 
\begin{align*}
\Mcal & =\{\Q\sim\P:S\text{ is a }\Q\text{-martingale}\}\\
\Mcal_{{\rm loc}} & =\{\Q\sim\P:S\text{ is a local }\Q\text{-martingale}\}\\
\Mcal_{{\rm sup}} & =\{\Q\sim\P:S\text{ is a }\Q\text{-supermartingale}\}\\
\Dcal_{{\rm loc}} & =\{Y:\text{ \ensuremath{Y}is càdlàg adapted, \ensuremath{Y\ge0}, \ensuremath{Y_{0}=1}, and \ensuremath{Y(1+H\cdot S)}is}\\
 & \qquad\qquad\text{a local martingale for every \ensuremath{1}-admissible \ensuremath{H}}\}\\
\Dcal_{{\rm sup}} & =\{Y:\text{ \ensuremath{Y}is càdlàg adapted, \ensuremath{Y\ge0}, \ensuremath{Y_{0}=1}, and \ensuremath{Y(1+H\cdot S)}is}\\
 & \qquad\qquad\text{a supermartingale for every constrained \ensuremath{1}-admissible \ensuremath{H}}\}.
\end{align*}
The elements of $\Mcal$ are called \emph{equivalent martingale measures}.
The elements of $\Mcal_{{\rm loc}}$ and $\Mcal_{{\rm sup}}$ are
called \emph{equivalent local martingale (supermartingale) measures},
respectively, whereas the elements of $\Dcal_{{\rm loc}}$ ($\Dcal_{{\rm sup}}$)
are called \emph{local martingale (supermartingale) deflators}. If
we identify a measure $\Q\sim\P$ with its Radon-Nikodym density process,
we have the following inclusions 
\[
\Mcal\subset\Mcal_{{\rm loc}}\subset\Mcal_{{\rm sup}}\subset\Dcal_{{\rm sup}}\quad\text{and}\quad\Mcal\subset\Mcal_{{\rm loc}}\subset\Dcal_{{\rm loc}}\subset\Dcal_{{\rm sup}}.
\]

The no arbitrage type conditions from Definition~\ref{D:arb} have
the following characterizations.

\begin{theorem} \label{T:NAchar} The following assertions hold. 
\begin{enumerate}
\item \label{T:NAchar:1} $\NUPBR$ ($\NUPBR_{C}$) holds if and only if
$\Dcal_{{\rm loc}}\ne\emptyset$ ($\Dcal_{{\rm sup}}\ne\emptyset$); 
\item \label{T:NAchar:2} $\NFLVR$ ($\NFLVR_{C}$) holds if and only if
$\Mcal_{{\rm loc}}\ne\emptyset$ ($\Mcal_{{\rm sup}}\ne\emptyset$); 
\item \label{T:NAchar:3} $\NFLVR$ and $\ND$ hold if and only if $\Mcal\ne\emptyset$. 
\end{enumerate}
\end{theorem}

\begin{proof} Part~{\ref{T:NAchar:1}} was proved in~\cite{Karatzas/Kardaras:2007},
Part~{\ref{T:NAchar:2}} in~\cite{Delbaen/Schachermayer:1994}
(see \cite{Karatzas/Kardaras:2007} for the constrained case), and
Part~{\ref{T:NAchar:3}} in~\cite{Jarrow:2012fk}. \end{proof}

\section{Informational Efficiency}

\label{S:eff}

We are interested in economies populated by agents who trade in the
risky and risk-free assets in order to maximize expected utility from
terminal wealth. We assume that the risky asset is available in unit
net supply, while the risk-free asset is available in zero net supply.
The agents face short sale restrictions on the risky asset, and therefore
solve optimization problems of the form 
\begin{equation}
u(x)=\sup_{H\in\Acal_{C}}\,\E\left[U(x+(H\cdot S)_{T})\right].\label{eq:primal}
\end{equation}
Here $x>0$ is the initial wealth, and $U:\Omega\times(0,\infty)\to\R$
is a {\em stochastic utility function}. By this we mean that $U$
is $\Fcal_{T}\otimes\Bcal(\R)$-measurable, $\E[|U(x)|]<\infty$ for
all $x>0$, and almost surely, $x\mapsto U(x)$ is continuously differentiable,
strictly increasing, strictly concave, with $\lim_{x\to\infty}U(x)=\infty$.
It also satisfies Inada conditions at zero and infinity: $\lim_{x\downarrow0}U'(x)=\infty$,
$\lim_{x\to\infty}U'(x)=0$. Its domain of definition is extended
to all of~$\R$ by setting $U(x)=-\infty$ for $x\le0$. We use the
convention that $\E[U(x+(H\cdot S)_{T})]=-\infty$ whenever the negative
part of $U(x+(H\cdot S)_{T})$ has infinite expectation, even if the
positive part also has infinite expectation.

The maximization problem~\eqref{eq:primal} has an associated family
of dual problems, 
\begin{equation}
v(y)=\sup_{Y\in\Dcal{\rm sup}}\,\E\left[V(yY_{T}))\right],\label{eq:dual}
\end{equation}
where $V(y)=\sup_{x>0}(U(x)-xy)$ is the conjugate of $U$, and $y>0$.
There is a large literature where existence and uniqueness of solutions
to \eqref{eq:primal}--\eqref{eq:dual} is established, and the optimizers
characterized, under a variety of assumptions on $U$ and $S$. Rather
than committing to a specific set of hypotheses, we will simply assume
that our utility functions are sufficiently well-behaved that the
main conclusions of the duality theory of \eqref{eq:primal}--\eqref{eq:dual}
are valid. The main reasons for doing this are to avoid imposing artificial
conditions on our utility functions, and to help emphasize which properties
are, in fact, important for the subsequent results.

\begin{definition} \label{D:viable} We say that a stochastic utility
function $U$ is {\em viable} if, whenever the conditions 
\begin{equation}
\Mcal_{{\rm sup}}\ne\emptyset,\qquad\text{\ensuremath{u(x)\in\R}for some \ensuremath{x>0},}\qquad\text{\ensuremath{v(y)\in\R}for some \ensuremath{y>0}}\label{eq:viable}
\end{equation}
hold, we have 
\begin{enumerate}
\item \label{D:viable:1} $u$ and $v$ are conjugate: $v(y)=\sup_{x>0}(u(x)-xy)$
and $u(x)=\inf_{y>0}(v(y)+xy)$. Furthermore, they are continuously
differentiable on~$(0,\infty)$. 
\item \label{D:viable:2} For each $x>0$, the primal problem~\eqref{eq:primal}
has an optimal solution $\widehat{H}\in\Acal_{C}$. If $H$ is another
optimal solution, then $(H\cdot S)_{T}=(\widehat{H}\cdot S)_{T}$. 
\item \label{D:viable:3} For each $y>0$, the dual problem~\eqref{eq:dual}
has a unique optimal solution $\widehat{Y}\in\Dcal_{{\rm sup}}$. 
\item \label{D:viable:4} If $x$ and $y$ are related via $y=u'(x)$, then
the corresponding optimal solutions satisfy 
\[
x+(\widehat{H}\cdot S)_{T}=I(y\widehat{Y}_{T}),
\]
where $I=(U')^{-1}$. Moreover, $\widehat{Y}(x+\widehat{H}\cdot S)$
is a martingale. 
\end{enumerate}
\end{definition}

Conditions under which viability holds are available in the extensive
literature on utility maximization. For example, Karatzas and Zitkovic
\cite{Karatzas/Zitkovic:2003} give conditions that are applicable
to the current setting (and allow the condition on $v$ in~\eqref{eq:viable}
to be dropped.) The seminal paper in the unconstrained case is Kramkov
anad Schachermayer \cite{Kramkov/Schachermayer:1999}, and a more
general characterization of viability is given in Mostovy~\cite{Mostovyi2012}.
The latter paper suggests that in general one should include the condition
on $v$ in~\eqref{eq:viable}.

Let us briefly remark on the possibility of having stochastic utility
functions. This is convenient for several reasons. First, it allows
us to cover our assumption of heterogeneous beliefs $\P^{*}\sim\P$
(replace $U(x)$ by $ZU(x)$, where $Z\d\P=\d\P^{*}$.) Second, the
initial discounting of the prices is without loss of generality (replace
$U(x)$ by $U(\exp(\int_{0}^{T}r_{t}\d t)x)$, where $r$ is the interest
rate.) Third, the assumption that utility depends on final wealth
can be replaced by the assumption that utility arises from consuming,
at time $T$, a good whose unit price is $\psi$ (replace $U(x)$
by $U(x\psi)$.) Of course, the above statements all come with the
caveat that viability must be preserved under the respective modifications
of the utilities.

We now consider our $n$ agents, indexed by $k=1,\ldots,n$. Each
agent is equipped with a viable stochastic utility function $U_{k}$
and some initial wealth $x_{k}>0$. In equilibrium, the agent will
follow some trading strategy $\widehat{H}^{k}\in\Acal_{C}$. We refer
to $(U_{k},x_{k},\widehat{H}^{k}:k=1,\ldots,n)$ as the {\em constrained
agent characteristics}. We use the following standard notion of an
economic equilibrium.

\begin{definition} \label{D:equil} A \emph{constrained equilibrium}
is a financial market $S$ together with constrained agent characteristics
$(U_{k},x_{k},\widehat{H}^{k}:k=1,\ldots,n)$ such that 
\begin{enumerate}
\item \label{D:equil:opt} individual optimality holds: For each $k$, $\widehat{H}^{k}$
is an optimal solution to~\eqref{eq:primal} with $U=U_{k}$ and
$x=x_{k}$, whose optimal value is finite. It is also required that
the dual value function is finite for some~$y$; 
\item \label{D:equil:clear} markets clear: $\widehat{H}_{t}^{1}+\cdots+\widehat{H}_{t}^{n}=1$
for all $t\in[0,T]$ and $x_{1}+\cdots+x_{n}=S_{0}$. 
\end{enumerate}
\end{definition}

The aggregate wealth in the economy at time zero is $S_{0}$, since
the risky asset is in unit net supply with initial price $S_{0}$,
and the risk-free asset is in zero net supply. This explains the form
of the market clearing condition for the initial wealth levels.

Note also that the market clearing condition {\ref{D:equil:clear}}
automatically implies that market clearing holds for the risk-free
asset. Indeed, the holdings of investor~$k$ in the risk-free asset
is $\widehat{H}_{t}^{0,k}=x_{k}+(\widehat{H}^{k}\cdot S)_{t}-\widehat{H}_{t}^{k}S_{t}$
by the self-financing property. Therefore, 
\[
{\textstyle \sum_{k}\widehat{H}_{t}^{0,k}=\sum_{k}x_{k}+(1\cdot S)_{t}-S_{t}=S_{0}+S_{t}-S_{0}-S_{t}=0.}
\]

We can now introduce the notion of {\em constrained informational
efficiency}. It generalizes the notion of informational efficiency
discussed in~\cite{Jarrow:2012fk} to the case where short sale restrictions
are present.

\begin{definition} \label{D:efficient} A financial market $S$ is
called {\em constrained informationally efficient} if it supported
by some constrained equilibrium, that is, if there exist constrained
agent characteristics $(U_{k},x_{k},\widehat{H}^{k}:k=1,\ldots,n)$
which together with $S$ form a constrained equilibrium. \end{definition}

Our aim is now to generate necessary conditions for constrained informational
efficiency. Sufficient conditions will be investigated in Section~\ref{S:repr},
where the stronger property of the existence of a representative agent
plays a key role.

\begin{lemma} \label{L:NFLVRC} Assume that for some $x>0$ the maximization
problem~\eqref{eq:primal} has an optimal solution $\widehat{H}$
with finite optimal value. Then 
\begin{enumerate}
\item \label{L:NFLVRC:1} $\widehat{H}$ is $C$-maximal; 
\item \label{L:NFLVRC:2} $\NFLVR_{C}$ holds. 
\end{enumerate}
\end{lemma}

\begin{proof} {\ref{L:NFLVRC:1}}: Let $K$ be a constrained admissible
strategy with $\Delta:=(K\cdot S)_{T}-(\widehat{H}\cdot S)_{T}\ge0$.
Since $\widehat{H}$ gives finite utility and since $U_{k}$ strictly
increasing, $K$ would yield a strict improvement if $\P(\Delta>0)>0$.
Hence $\Delta=0$, and $\widehat{H}$ is $C$-maximal.

{\ref{L:NFLVRC:2}}: We need to establish $\NA_{C}$ and $\L^{0}$-boundedness
of the set $\Kcal_{C}$ in Definition~\ref{D:arb}. The argument
is well-known. First, suppose $K$ is a constrained admissible strategy
with $(K\cdot S)_{T}\ge0$. Then $\widehat{H}+K$ is also constrained
admissible, and we have $((\widehat{H}+K)\cdot S)_{T}\ge(\widehat{H}\cdot S)_{T}$.
But $\widehat{H}$ is $C$-maximal by Part~$(i)$, so this inequality
must be an equality. Hence $(K\cdot S)_{T}=0$, and we deduce~$\NA_{C}$.
To prove that $\Kcal_{C}$ is bounded in~$\L^{0}$, we adapt~\cite[Proposition~4.19]{Karatzas/Kardaras:2007}
to the case with stochastic utility functions. Assume for contradiction
that $\Kcal_{C}$ is not bounded in~$\L^{0}$. Then we can find $\varepsilon>0$
and constrained $1$-admissible strategies $H^{n}$, $n\in\N$, such
that $\P(1+(H^{n}\cdot S)_{T}>n)>\varepsilon$. Thus, 
\[
u(\varepsilon+1)\ge\E\left[U(\varepsilon+1+(H^{n}\cdot S)_{T})\right]\ge-\E\left[|U(\varepsilon)|\right]+\left[U(\varepsilon+1+n)\1{1+(H^{n}\cdot S)_{T}>n}\right].
\]
Since $U(\varepsilon+1+n)\to\infty$ almost surely, we can find, for
any $m>0$, some large $n$ such that $\P(U(\varepsilon+1+n)>m)>1-\varepsilon/2$.
Then, 
\[
U(\varepsilon+1+n)\1{1+(H^{n}\cdot S)_{T}>n}\ge m\1{1+(H^{n}\cdot S)_{T}>n\text{ and }U(\varepsilon+1+n)>m}-|U(\varepsilon)|,
\]
so that 
\begin{align*}
u(\varepsilon+1) & \ge-2\,\E\left[|U(\varepsilon)|\right]+m\P\left(1+(H^{n}\cdot S)_{T}>n\text{ and }U(\varepsilon+1+n)>m\right)\\
 & \ge-2\,\E\left[|U(\varepsilon)|\right]+m\frac{\varepsilon}{2}.
\end{align*}
It follows that $u(\varepsilon+1)=\infty$, which by concavity of
$u$ implies $u(x)=\infty$. This is the desired contradiction. \end{proof}

We are now ready to give our first main result, which states that
in a constrained equilibrium there is in fact a local martingale deflator.
\emph{A priori} one would only expect a supermartingale deflator to
exist. However, market clearing implies that at every point in time,
there must be some investor who optimally holds a positive number
of shares in the risky asset. Such an investor is locally unconstrained,
and this forces the price to behave, in a qualitative way, as it would
in an unconstrained economy---otherwise the investor would reduce
his holdings, contradicting optimality.

\begin{theorem} \label{T:NUPBR} Let $(S;\ U_{k},x_{k},\widehat{H}^{k}:k=1,\ldots,n)$
be a constrained equilibrium. Then $\Dcal_{{\rm loc}}$ is non-empty.
\end{theorem}

It is a striking fact that the above heuristic motivation for Theorem~\ref{T:NUPBR}
fails if we are interested in the $\NA$ property rather than the
$\NUPBR$ property. Indeed, in Section~\ref{S:ex2} we provide an
example where both $\NFLVR_{C}$ and $\NUPBR$ hold, $\NA$ fails,
and a strictly positive strategy $\widehat{H}$ (in fact, the constant
strategy $\widehat{H}\equiv1$) is optimal for~\eqref{eq:primal}.
In other words, even though the optimal strategy is strictly positive,
the short sale constraint is binding in the sense that an unconstrained
agent would choose a very different strategy.

An immediate consequence of Theorem~\ref{T:NUPBR} and Lemma~\ref{L:NFLVRC}{\ref{L:NFLVRC:2}}
is the following.

\begin{corollary} A constrained informationally efficient market~$S$
necessarily satisfies both $\NFLVR_{C}$ and $\NUPBR$. \end{corollary}

The proof of Theorem~\ref{T:NUPBR} relies on the following lemma.

\begin{lemma} \label{L:doptlm} Let $U$ be a viable stochastic utility
function. Suppose \eqref{eq:viable} holds, and let $\widehat{Y}$
be the optimal solution to the dual problem~\eqref{eq:dual} for
some $y>0$. Then $\widehat{Y}$ is a strictly positive local martingale.
\end{lemma}

\begin{proof} The Inada conditions imply $I(0)=\infty$, which forces
$\widehat{Y}_{T}>0$ due to {\ref{D:viable:4}} in the definition
of viability. Since $\widehat{Y}$ is a nonnegative supermartingale
this yields~$\widehat{Y}>0$, proving the first assertion. As a consequence
(and since $\widehat{Y}_{0}=1$), there is a local martingale $N$
and a nondecreasing predictable process~$B$ such that $\widehat{Y}=\Ecal(-N-B)$.
We need to prove that $B=0$. To this end, observe that since $S>0$
there is a continuous semimartingale $M+A$, with $M$ a local martingale
and $A$ a predictable finite variation process, such that $S=S_{0}\Ecal(M+A)$.
Furthermore, for any nonnegative predictable $(M+A)$-integrable process
$\theta$, $X=\Ecal(\theta\cdot(M+A))$ is the wealth process of the
constrained $1$-admissible strategy $H=S^{-1}\theta X$. Indeed,
we have $H\ge0$, $X>0$, and 
\[
X=1+X\cdot(\theta\cdot(M+A))=1+(X\theta S^{-1})\cdot(S\cdot(M+A))=1+H\cdot S.
\]
Since $\widehat{Y}\in\Dcal_{{\rm sup}}$, therefore, $\widehat{Y}X$
is a supermartingale. Moreover, Yor's formula and the fact that $[M,B]=[A,B]=[A,N]=0$
due to the continuity of $M$ and $A$, yield 
\[
\widehat{Y}X=\Ecal(-N-B)\Ecal(\theta\cdot(M+A))=\Ecal(-N+\theta\cdot M+\theta\cdot(A-[N,M])-B).
\]
It follows that $\theta\cdot(A-[N,M])-B$ is a non-increasing process
for any $\theta\ge0$, and thus $A-[N,M]$ must also be non-increasing.
We claim that this implies that $Y'=\Ecal(-N)$ lies in $\Dcal_{{\rm sup}}$.
To see why, note that for any constrained $1$-admissible $H$ we
have, with $X=1+H\cdot S$, 
\begin{align*}
Y'X & =1+Y'_{-}\cdot X+X\cdot Y'+[Y',X]\\
 & =1+(Y'_{-}HS)\cdot(M+A)+X\cdot Y'-(Y'_{-}HS)\cdot[N,M]\\
 & =1+(Y'_{-}HS)\cdot M+X\cdot Y'+(Y'_{-}HS)\cdot(A-[N,M]).
\end{align*}
Only the last term is not a local martingale, and it is non-increasing
since $Y'_{-}HS\ge0$ and $A-[N,M]$ is non-increasing. So $Y'\in\Dcal_{{\rm sup}}$
as claimed. But we also have 
\[
\widehat{Y}_{T}=\Ecal(-N-B)_{T}\le\Ecal(-N)_{T}=Y'_{T},
\]
with strict inequality on the set $\{B_{T}>0\}$. If this set had
positive probability, $Y'$ would achieve a strictly smaller objective
value than $\widehat{Y}$ in the dual problem~\eqref{eq:dual}, contradicting
the optimality of $\widehat{Y}$. Hence $B_{T}=0$, and the lemma
is proved. \end{proof}

\begin{proof}[Proof of Theorem~\ref{T:NUPBR}] Let $\widehat{Y}^{k}$
be the dual optimizer of the $k$:th investor's problem. By Lemma~\ref{L:doptlm}
it is a strictly positive local martingale, and condition {\ref{D:viable:4}}
in the definition of viability implies that $\widehat{Y}^{k}(x_{k}+\widehat{H}^{k}\cdot S)$
is a martingale. In particular we may write $\widehat{Y}^{k}=\Ecal(-N^{k})$
for some local martingale $N^{k}$. Letting $S=M+A$ be the canonical
decomposition of $S$, integration by parts yields 
\begin{align*}
\widehat{Y}^{k}(x_{k}+\widehat{H}^{k}\cdot S) & =x_{k}+(x_{k}+\widehat{H}^{k}\cdot S)\cdot\widehat{Y}_{-}^{k}+(\widehat{Y}^{k}\widehat{H}^{k})\cdot(M+A)+\widehat{H}^{k}\cdot[\widehat{Y}^{k},M]\\
 & =x_{k}+(x_{k}+\widehat{H}^{k}\cdot S)\cdot\widehat{Y}_{-}^{k}+(\widehat{Y}_{-}^{k}\widehat{H}^{k})\cdot M+(\widehat{H}^{k}\widehat{Y}_{-}^{k})\cdot(A-[N^{k},M]).
\end{align*}
Since the left side is a $\P$-martingale, and since $\widehat{Y}^{k}>0$,
we deduce $\widehat{H}^{k}\cdot(A-[N^{k},M])=0$, and hence 
\begin{equation}
\1{\widehat{H}^{k}>0}\cdot(A-[N^{k},M])=0.\label{eq:kpos}
\end{equation}
We now define predictable sets $D_{k}$ by 
\[
D_{1}=\{\widehat{H}^{1}>0\},\qquad D_{k}=\{\widehat{H}^{k}>0\}\setminus D_{k-1},\quad k=2,\ldots,n.
\]
Then clearly $D_{k}\cap D_{\ell}=\emptyset$ whenever $k\ne\ell$.
Moreover, market clearing implies that for every $(t,\omega)$, there
is at least one $k\in\{1,\ldots,n\}$ such that $\widehat{H}_{t}^{k}(\omega)>0$.
Therefore $\cup_{k}D_{k}=\cup_{k}\{\widehat{H}^{k}>0\}=[0,T]\times\Omega$,
and hence 
\[
\oo_{D_{1}}+\cdots+\oo_{D_{n}}=1.
\]
Finally, since $D_{k}\subset\{\widehat{H}^{k}>0\}$, it follows from~\eqref{eq:kpos}
that $\oo_{D_{k}}\cdot(A-[N^{k},M])=0$ for each~$k$. Thus with
\[
N=\oo_{D_{1}}\cdot N^{1}+\cdots\oo_{D_{n}}\cdot N^{n},
\]
we get 
\[
A-[N,M]=\sum_{k=1}^{n}\oo_{D_{k}}\cdot(A-[N^{k},M])=0.
\]
Setting $Y=\Ecal(-N)$ this yields $YS=S\cdot Y+Y\cdot M$, which
is a local martingale. Thus $Y\in\Dcal_{{\rm loc}}$, and the theorem
is proved. \end{proof}

\section{Representative Agent Equilibria}

\label{S:repr}

In this section we consider financial markets that are not only informationally
efficient in the sense of Definition~\ref{D:efficient}, but have
the additional feature that a representative agent can be constructed.
Such markets turn out to have a simple dual characterization in terms
of supermartingale measures and local martingale deflators, see Theorem~\ref{T:reprchar}
below. We also consider the question of aggregation of a multi-agent
equilibrium to a representative agent equilibrium. This issue has
a long history in the literature, see for instance Magill\cite{Magill1981},
Constantinides \cite{Constantinides1982}, Karatzas, Lehoczky, and
Shreve \cite{Karatzas/Lehoczky/Shreve1990}, Cuoco and He \cite{CuocoHe2001}.

\begin{definition} A {\em constrained representative agent equilibrium}
$(S;U)$ is a constrained equilibrium with one single agent. That
is, the constrained agent characteristics is $(U,1,1)$. \end{definition}

By market clearing, a representative agent in any equilibrium (constrained
or not) must use the constant strategy $\widehat{H}\equiv1$. This
strategy automatically lies in $\Acal_{C}$, clearly satisfying the
short sale constraint. However, there exist constrained representative
agent equilibria, where the constant strategy is optimal in the class
$\Acal_{C}$, but not in the class $\Acal$ of unconstrained admissible
strategies. An example is given in Section~\ref{S:ex2}.

The following is the main result of this section.

\begin{theorem} \label{T:reprchar} Consider a financial market $S$.
The following conditions are equivalent: 
\begin{enumerate}
\item \label{T:reprchar:1} $(S,U)$ is a constrained representative agent
equilibrium for some viable stochastic utility function $U$, 
\item \label{T:reprchar:2} $\Mcal_{{\rm sup}}\ne\emptyset$ and there exists
a process $Z\in\Dcal_{{\rm loc}}$ such that $ZS$ is a martingale. 
\end{enumerate}
In either case, $S$ satisfies $\NFLVR_{C}$, $\ND_{C}$, and $\NUPBR$.
\end{theorem}

An immediate consequence of this theorem and the definition of constrained
informational efficiency is the following.

\begin{corollary} If $\Mcal_{{\rm sup}}\ne\emptyset$ and there exists
a process $Z\in\Dcal_{{\rm loc}}$ such that $ZS$ is a martingale,
then the market is constrained informationally efficient. \end{corollary}

The following lemma is used in the proof of Theorem~\ref{T:reprchar}.
It essentially shows that viability is invariant under changes of
beliefs. See Källblad \cite[Theorem~1]{Kallblad2013} for a related
result.

\begin{lemma}\label{L:viableinv} Let $U^{*}$ be a viable stochastic
utility function, and $\P^{*}\sim\P$ an equivalent probability measure.
Set $Z_{t}=\frac{\d\P^{*}}{\d\P}|_{\Fcal_{t}}$. Then the stochastic
utility function $U=Z_{T}U^{*}$ is again viable. \end{lemma}

\begin{proof} Since the class $\Acal_{C}$ is invariant under equivalent
changes of probability measure, we have 
\begin{equation}
u(x)=\sup_{H\in\Acal_{C}}\E[U(x+(H\cdot S)_{T})]=\sup_{H\in\Acal_{C}}\E^{*}[U^{*}(x+(H\cdot S)_{T})],\label{eq:ustar}
\end{equation}
where $\E^{*}[\,\cdot\,]$ denotes expectation under $\P^{*}$. Furthermore,
we have 
\[
V(y)=\sup_{x>0}\,\left(U(x)-xy\right)=Z_{T}\sup_{x>0}\,\left(U^{*}(x)-x\frac{y}{Z_{T}}\right)=Z_{T}V^{*}\left(\frac{y}{Z_{T}}\right),
\]
where $V^{*}$ is the conjugate of $U^{*}$. Therefore, 
\begin{equation}
v(y)=\inf_{Y\in\Dcal_{{\rm sup}}}\E\left[Z_{T}V^{*}\left(\frac{yY_{T}}{Z_{T}}\right)\right]=\inf_{Y^{*}\in\Dcal_{{\rm sup}}^{*}}\E^{*}\left[V^{*}\left(yY_{T}^{*}\right)\right],\label{eq:vstar}
\end{equation}
where $\Dcal_{{\rm sup}}^{*}=\{Y/Z:Y\in\Dcal_{{\rm sup}}\}$. Note
that this is the set of all càdlàg adapted $Y^{*}$ with $Y^{*}\ge0$
and $Y_{0}^{*}=1$ such that $Y^{*}(1+H\cdot S)$ is a $\P^{*}$-supermartingale
for every constrained 1-admissible~$H$. Suppose now~\eqref{eq:viable}
holds, and note that this condition is invariant under equivalent
changes of measure. In view of~\eqref{eq:ustar} and~\eqref{eq:vstar},
the viability of $U^{*}$ implies that $u$ and $v$ are conjugate
and continuously differentiable. Moreover, the primal problem has
an optimal solution~$\widehat{H}\in\Acal_{C}$, and the dual problem
an optimal solution $\widehat{Y}^{*}\in\Dcal_{{\rm sup}}^{*}$. Both
are unique in the appropriate sense, see Definition~\ref{D:viable}{\ref{D:viable:2}}--{\ref{D:viable:3}}.
The process $\widehat{Y}=Z\widehat{Y}^{*}$ then lies in $\Dcal_{{\rm sup}}$
and is optimal for the dual problem under~$\P$. At this point we
have established properties {\ref{D:viable:1}}--{\ref{D:viable:3}}
in the definition of viability. Property~{\ref{D:viable:4}} follows
from the equality 
\[
x+(\widehat{H}\cdot S)_{T}=I^{*}(yY_{T}^{*})=I(yZ_{T}Y_{T}^{*})=I(yY_{T})
\]
and the fact that $\widehat{Y}(x+\widehat{H}\cdot S)$ is a $\P$-martingale
if and only if $\widehat{Y}^{*}(x+\widehat{H}\cdot S)$ is a $\P^{*}$-martingale.
The lemma is proved. \end{proof}

\begin{proof}[Proof of Theorem~\ref{T:reprchar}] {\ref{T:reprchar:1}}
$\Longrightarrow$ {\ref{T:reprchar:2}}: Lemma~\ref{L:NFLVRC}
implies that $\NFLVR_{C}$ holds, so that $\Mcal_{{\rm sup}}\ne\emptyset$
by Theorem~\ref{T:NAchar}. Market clearing together with viability
of $U$ show that the dual optimizer $Z=\widehat{Y}$ has the property
that $ZS$ is a martingale. Moreover, $Z$ is a local martingale by
Lemma~\ref{L:doptlm}. This implies that $Y\in\Dcal_{{\rm loc}}$.
Indeed, let $H\in\Acal_{C}$ be 1-admissible. Integration by parts
applied twice (and associativity of the stochastic integral) yields
\begin{align*}
Z(1+H\cdot S) & =1+(1+H\cdot S)\cdot Z+Z_{-}\cdot(H\cdot S)+H\cdot[Z,S]\\
 & =1+(1+H\cdot S-HS)\cdot Z+H\cdot(ZS).
\end{align*}
Thus the left side is a stochastic integral with respect to the (two-dimensional)
local martingale $(Z,ZS)$. Since it is nonnegative it is again a
local martingale by the Ansel-Stricker theorem.

{\ref{T:reprchar:2}} $\Longrightarrow$ {\ref{T:reprchar:1}}:
The proof is inspired by \cite[Theorem~3.2]{Jarrow:2012fk}. As a
candidate utility function for the representative agent we take 
\[
U(x)=\frac{Z_{T}S_{T}^{\gamma}x^{1-\gamma}}{1-\gamma}
\]
for some $\gamma\in(0,1)$. Let us show that $\widehat{H}\equiv1$
is optimal. Clearly $\widehat{H}\in\Acal_{C}$. Next, the corresponding
utility is $\E[U(S_{T})]=\E[Z_{T}S_{T}]/(1-\gamma)=1/(1-\gamma)$,
in particular it is finite. Now let $H\in\Acal_{C}$ be arbitrary.
We may assume that the final wealth $1+(H\cdot S)_{T}$ is positive
almost surely, otherwise the utility would be $-\infty$. This implies
that the wealth process is in fact 1-admissible. Indeed, $H\cdot S$
is a supermartingale under any $\Q\in\Mcal_{{\rm sup}}$, which exists
by assumption, so 
\begin{equation}
1+(H\cdot S)_{t}\ge\E^{\Q}[1+(H\cdot S)_{T}\mid\Fcal_{t}]>0.\label{eq:wealthpos}
\end{equation}
Here we used \cite[Proposition~3.5]{Pulido2014}, the ``Ansel-Stricker
theorem for supermartingales''. It is now easy to see that $\widehat{H}$
is optimal. Indeed, concavity of $U$, the equality $U'(S_{T})=Z_{T}$,
and the properties of $Z$ yield 
\begin{align*}
\E[U(1+(H\cdot S)_{T})] & \le\E[U(S_{T})]+\E[U'(S_{T})(1+(H\cdot S)_{T})]-\E[U'(S_{T})S_{T}]\\
 & =\E[U(S_{T})]+\E[Z_{T}(1+(H\cdot S)_{T})]-\E[Z_{T}S_{T}]\\
 & \le\E[U(S_{T})].
\end{align*}
We need to prove that $U$ is viable. To this end, write $U(x)=\frac{\d\P^{*}}{\d\P}U^{*}(x)$,
where 
\[
\frac{\d\P^{*}}{\d\P}=\frac{Z_{T}S_{T}^{\gamma}}{\E[Z_{T}S_{T}^{\gamma}]},\qquad U^{*}(x)=\E[Z_{T}S_{T}^{\gamma}]\frac{x^{1-\gamma}}{1-\gamma}.
\]
Note that $\E[Z_{T}S_{T}^{\gamma}]\le\E[Z_{T}(1+S_{T})]<\infty$.
The utility function $U^{*}$ is viable by \cite[Theorem~3.10]{Karatzas/Zitkovic:2003},
so Lemma~\ref{L:viableinv} implies that $U$ is also viable.

It now only remains to prove that one (and hence both) of {\ref{T:reprchar:1}}
and {\ref{T:reprchar:2}} implies $\NFLVR_{C}$, $\ND_{C}$, and
$\NUPBR_{C}$. In view of what we have already done, only $\ND_{C}$
needs to be verified. But this follows directly from Lemma~\ref{L:NFLVRC}{\ref{L:NFLVRC:1}}.
\end{proof}

In the complete market setting one expects that a representative agent
equilibrium can be found. The following result confirms this intuition.

\begin{proposition} \label{P:complete} Let $(S;\ U_{k},x_{k},\widehat{H}^{k}:k=1,\ldots,n)$
be a constrained equilibrium. If $\Dcal_{{\rm loc}}$ is a singleton,
and $\widehat{H}^{k}>0$ for all $k$, then the two equivalent conditions
of Theorem~\ref{T:reprchar} hold. \end{proposition}

\begin{proof} Let $\widehat{Y}^{k}$ be the dual optimizers for the
agents' optimization problems. By Lemma~\ref{L:doptlm} there exist
local martingales $N^{k}$ such that $\widehat{Y}^{k}=\Ecal(-N^{k})$.
Let $S=M+A$ be the canonical decomposition of $S$. Integration by
parts yields 
\[
\widehat{Y}^{k}(x_{k}+\widehat{H}^{k}\cdot S)=(x_{k}+\widehat{H}^{k}\cdot S)\cdot\widehat{Y}^{k}+\widehat{Y}_{-}^{k}\widehat{H}^{k}\cdot(M+A-[N^{k},M]).
\]
Since $\widehat{Y}^{k}(x_{k}+\widehat{H}^{k}\cdot S)$ is a martingale,
and since $\widehat{Y}_{-}^{k}\widehat{H}^{k}>0$, we get $A=[N^{k},M]$.
Another application of the integration by parts formula shows that
$\widehat{Y}^{k}S$ is a local martingale, so we deduce $\widehat{Y}^{k}\in\Dcal_{{\rm loc}}=\{Z\}$
and hence $\widehat{Y}^{k}=Z$ for all $k$. Market clearing gives
\[
ZS=Z\sum_{k=1}^{n}(x_{k}+\widehat{H}^{k}\cdot S)=\sum_{k=1}^{n}\widehat{Y}^{k}(x_{k}+\widehat{H}^{k}\cdot S),
\]
which is a (true) martingale. Since also $\Mcal_{{\rm sup}}\ne\emptyset$
by Lemma~\ref{L:NFLVRC}, condition~{\ref{T:reprchar:2}} of Theorem~\ref{T:reprchar}
is satisfied. \end{proof}

Theorem~\ref{T:reprchar} characterizes those price processes~$S$
for which a supporting representative agent equilibrium exists. However,
the choice of representative agent utility $U$ is far from unique
in general, and it is natural to ask if some choices are more natural
(or useful) than others. In particular, if the price process $S$
is known to come from a constrained equilibrium $(S;\ U_{k},x_{k},\widehat{H}^{k}:k=1,\ldots,n)$,
it is of interest to know whether $U_{1},\ldots,U_{n}$ can be {\em
aggregated} to form a representative utility $U$. This was first
done in Negishi \cite{Negishi1960} in the complete market case, see
also Magill\cite{Magill1981}, Constantinides \cite{Constantinides1982},
Karatzas, Lehoczky, and Shreve \cite{Karatzas/Lehoczky/Shreve1990},
and in Cuoco and He \cite{CuocoHe2001} for incomplete markets. Extensions
involving investment constraints have been considered in Basak and
Cuoco \cite{BasakCuoco1998} and Hugonnier \cite{Hugonnier2012}.
Specifically, one considers functions 
\begin{equation}
U(x;\lambda)=\max_{c_{1}+\cdots+c_{n}=x}\ \sum_{k=1}^{n}\lambda_{k}U_{k}(c_{k}),\label{eq:Uaggr}
\end{equation}
for weight vectors $\lambda=(\lambda_{1},\ldots,\lambda_{n})\in\R_{++}^{n}$,
possibly stochastic, and asks for a weight vector such that the constant
strategy is optimal for $U(\,\cdot\,;\lambda)$ given the prevailing
price process~$S$. In this case we call $U(\,\cdot\,;\lambda)$
an {\em aggregate utility function}.

\begin{theorem} \label{T:aggr} Let $(S;\ U_{k},x_{k},\widehat{H}^{k}:k=1,\ldots,n)$
be a constrained equilibrium supporting a representative agent---that
is, $S$ satisfies one of the equivalent conditions in Theorem~\ref{T:reprchar}.
Then there are stochastic weights $\lambda=(\lambda_{1},\ldots,\lambda_{n})$
such that $U(\,\cdot\,;\lambda)$ is an aggregate utility function.
\end{theorem}

\begin{proof} Fix any weights $\lambda=(\lambda_{1},\ldots,\lambda_{n})$
and any $x>0$. The Lagrangian corresponding to the optimization problem
in~\eqref{eq:Uaggr} is 
\[
L(c_{1},\ldots,c_{n};\mu)=\sum_{k=1}^{n}\lambda_{k}U_{k}(c_{k})+\mu\Big(x-\sum_{k=1}^{n}c_{k}\Big),
\]
where $\mu\in\R$ is the Lagrange multiplier. By strict concavity,
the maximizer $(c_{1}^{*},\ldots,c_{n}^{*})$ and associated Lagrange
multiplier $\mu^{*}$ are characterized by the first order conditions,
namely 
\begin{equation}
x=c_{1}^{*}+\cdots+c_{n}^{*},\qquad\lambda_{k}U_{k}'(c_{k}^{*})-\mu^{*}=0\quad(k=1,\ldots,n).\label{eq:LFOC}
\end{equation}
Now, let $Z\in\Dcal_{{\rm loc}}$ be such that $ZS$ is a martingale---such
$Z$ exists by hypothesis. Moreover, let $\widehat{X}_{T}^{k}=x_{k}+(\widehat{H}^{k}\cdot S)_{T}$
be the final wealth of agent~$k$, and set 
\[
\lambda_{k}=\frac{Z_{T}}{U_{k}'(\widehat{X}_{T}^{k})},\qquad k=1,\ldots,n.
\]
Since also $\widehat{X}_{T}^{1}+\cdots+\widehat{X}_{T}^{n}=S_{T}$
holds by market clearing, \eqref{eq:LFOC} is satisfed with $c_{k}^{*}=\widehat{X}_{T}^{k}$
and $\mu^{*}=Z_{T}$. Hence 
\begin{equation}
U(S_{T};\lambda)=\lambda_{1}U_{1}(\widehat{X}_{T}^{1})+\cdots+\lambda_{n}U_{n}(\widehat{X}_{T}^{n}).\label{eq:LFOC2}
\end{equation}
Consider an arbitrary $H\in\Acal_{C}$ with strictly positive final
value $X_{T}=1+(H\cdot S)_{T}$. Then, for some $X_{T}^{1},\ldots,X_{T}^{n}$
with $X_{T}^{1}+\cdots+X_{T}^{n}=X_{T}$, we have 
\begin{align*}
U(X_{T};\lambda) & =\lambda_{1}U_{1}(X_{T}^{1})+\cdots+\lambda_{n}U_{n}(X_{T}^{n})\\
 & \le\sum_{k=1}^{n}\left[\lambda_{k}U_{k}(\widehat{X}_{T}^{k})+\lambda_{k}U_{k}'(\widehat{X}_{T}^{k})(X_{T}^{k}-\widehat{X}_{T}^{k})\right]\\
 & =U(S_{T};\lambda)+\sum_{k=1}^{n}\lambda_{k}U_{k}'(\widehat{X}_{T}^{k})(X_{T}^{k}-\widehat{X}_{T}^{k})\\
 & =U(S_{T};\lambda)+Z_{T}(X_{T}-S_{T}),
\end{align*}
where we used the concavity of $U_{k}$, the identity~\eqref{eq:LFOC2},
as well the fact that $Z_{T}=\lambda_{k}U_{k}'(\widehat{X}_{T}^{k})$
for all $k$. Using that $X$ is strictly positive by $\NFLVR_{C}$
(see~\eqref{eq:wealthpos} in the proof of Theorem~\ref{T:reprchar}),
so that $ZX$ is a positive local martingale and thus a supermartingale,
we get 
\[
\E[U(X_{T};\lambda)]\le\E[U(S_{T};\lambda)]+\E[Z_{T}(X_{T}-S_{T})]\le\E[U(S_{T};\lambda)].
\]
Thus the constant strategy is optimal for $U(\,\cdot\,;\lambda)$,
which is what we had to prove. \end{proof}

\section{Examples}

\label{S:ex}

In this section we provide examples illustrating to what extent our
previous results are sharp. The examples are based on Proposition~\ref{P:counterex1}
below, which is also interesting in its own right. We work on the
canonical path space $(\Omega,\Fcal,\F,\P)$, where the coordinate
process $W$ is Brownian motion under $\P$, and $\F$ is its natural
filtration. Some care is needed in the description of $\Omega$: it
is the set of all functions $\omega:\R_{+}\to\R\cup\{\infty\}$ that
are continuous on $[0,\zeta(\omega))$, where $\zeta(\omega)=\inf\{t\ge0:\omega(t)=\infty\}$,
and satisfy $\omega(t)=\infty$ for all $t\ge\zeta(\omega)$. Moreover,
the filtration $\F$ is not augmented with the $\P$-nullsets, and
this does not affect the stochastic calculus used below.

By means of a construction due to Föllmer~\cite{Follmer1972}, in
turn inspired by Doob's $h$-transform, any strictly positive local
martingale $Z$ with $Z_{0}=1$ can be viewed as the density process
of some probability measure, possibly not equivalent with respect
to~$\P$. Specifically, define its explosion time by 
\[
\sigma=\lim_{n}\sigma_{n},\qquad\sigma_{n}=\inf\left\{ t\ge0:Z_{t}\ge n\right\} ,
\]
which of course satisfies $\P(\sigma<\infty)=0$. Then there exists
a probability measure $\P^{*}$ on $\Fcal$ with $\P|_{\Fcal_{t}}\ll\P^{*}|_{\Fcal_{t}}$
for each $t\ge0$, such that 
\[
\frac{\d\P}{\d\P^{*}}\bigg|_{\Fcal_{t}}=\frac{1}{Z_{t}}\1{t<\sigma},
\]
and the equality $\E\left[Z_{t}\right]=\P^{*}(\sigma>t)$ holds for
each $t\ge0$. In particular, $Z$ is a true martingale under~$\P$
if and only if $\P^{*}(\sigma=\infty)=1$. We refer to $\P^{*}$ as
the {\em Föllmer measure} corresponding to~$Z$. The Föllmer measure
is uniquely determined on $\Fcal_{\sigma}$, but may have several
different extensions to $\Fcal$. This possible non-uniqueness plays
a key role in Larsson\cite{Larsson2013filtr} and is also discussed
in Perkowski and Ruf \cite{PerkowskiRuf2013}. For us the particular
choice of extension is not important.

Girsanov's theorem also extends to this setting. We will only need
the following simple version. If the local martingale $Z$ is of the
form $Z=\Ecal(\theta\cdot W)$ for some predictable $W$-integrable
process~$\theta$, there is a $\P^{*}$-Brownian motion $W^{*}$,
possibly defined on an extension of the original probability space,
such that 
\begin{equation}
W_{t}^{*}=W_{t}-\int_{0}^{t}\theta_{s}\d s,\qquad t<\sigma.\label{eq:Wstar}
\end{equation}
Indeed, Girsanov's theorem shows that $W_{t\wedge\sigma_{n}}^{*}$
is stopped Brownian motion under the equivalent measure $\P^{n}$
with density process $Z_{t\wedge\sigma_{n}}$ (this coincides with
$\P^{*}$ on $\Fcal_{\sigma_{n}}$), so $W_{t}^{*}$ is well-defined
on $\lc0,\sigma\lc$. Using that $\langle W^{*},W^{*}\rangle_{t}=t$,
one shows that $\lim_{t\uparrow\sigma}W_{t}^{*}$ exists $\P^{*}$-a.s.~on
$\{\sigma<\infty\}$. Thus, letting $W^{**}$ be an independent Brownian
motion (this is where we may have to extend the probability space),
the process $W^{*}\oo_{\lc0,\sigma\lc}+(W_{t}^{**}-W_{\sigma}^{**}+W_{\sigma}^{*})\oo_{\lc\sigma,\infty\lc}$
is Brownian motion under $\P^{*}$ and satisfies~\eqref{eq:Wstar}.
For further details regarding the construction of the Föllmer measure
and Girsanov's theorem, see Yoeurp\cite{Yoeurp1976}.

The following proposition is the basis for our examples below. It
is inspired by examples by Kazamaki, see~\cite[Chapter~1.4]{Kazamaki1994}.
To state the result, define 
\[
p(T)=\P\Big(\inf_{0\le t\le T}W_{t}\le-1\Big),
\]
the probability that standard Brownian motion hits $-1$ before time~$T$.

\begin{proposition} \label{P:counterex1} Fix $T>0$ and a constant
$\beta>1$. Let $X$ be the unique strong solution to 
\[
\d X_{t}=-X_{t}^{2}\d W_{t},\qquad X_{0}=1,
\]
and define a process $L$ and stopping time $\tau$ by 
\[
L=\frac{\Ecal(-\beta X\cdot W)}{\Ecal(-X\cdot W)}\qquad\text{and}\qquad\tau=\inf\left\{ t\ge0:L_{t}\ge1+p(T)^{-1}\right\} .
\]
Then the local martingales $Z^{(1)}=\Ecal(-X\oo_{\lc0,\tau\rc}\cdot W)$
and $Z^{(\beta)}=\Ecal(-\beta X\oo_{\lc0,\tau\rc}\cdot W)$ satisfy
$\E[Z_{T}^{(1)}]<1$ and $\E[Z_{T}^{(\beta)}]=1$. \end{proposition}

The idea of the construction is best understood by passing to the
Föllmer measures $\P^{(1)}$ and $\P^{(\beta)}$ corresponding to
$Z^{(1)}$ and $Z^{(\beta)}$, respectively. It is not hard to see
that if $\tau$ were always infinite, both these density processes
would explode under their respective Föllmer measure, and would thus
both be strict local martingales under~$\P$. Of course $\tau$ is
not always infinite---in fact, it is designed to occur early enough
to $\P^{(\beta)}$-almost surely prevent $Z^{(\beta)}$ from exploding,
while at the same time occurring late enough that $Z^{(1)}$ does
explode with positive $\P^{(1)}$-probability. The device for achieving
this is the ``likelihood ratio'' process $L$ (note that we have
$L_{t}=Z_{t}^{(\beta)}/Z_{t}^{(1)}$ for $t<\tau$.) It turns out
that $L$ is a supermartingale under $\P^{(1)}$, and hence fails
to trigger $\tau$ with positive probability, but is an exploding
submartingale under $\P^{(\beta)}$ (or, more precisely, would have
exploded without the intervention of the stopping time~$\tau$.)
We now turn to the details.

\begin{proof} We first prove $\E[Z_{T}^{(1)}]<1$. Let $\P^{(1)}$
be the Föllmer measure associated with $Z^{(1)}$, and let $W^{(1)}$
be a $\P^{(1)}$-Brownian motion such that 
\[
W_{t}^{(1)}=W_{t}+\int_{0}^{t}X_{s}\1{s<\tau}\d s,\qquad t<\sigma^{(1)},
\]
where $\sigma^{(1)}$ is the explosion time of $Z^{(1)}$. A calculation
yields, for all $t<\sigma^{(1)}$, the equalities $X_{t\wedge\tau}=\Ecal(-X\cdot W)_{t\wedge\tau}=Z_{t}^{(1)}$
and $\frac{1}{X_{t\wedge\tau}}=1+W_{t\wedge\tau}+\int_{0}^{t\wedge\tau}X_{s}\d s=1+W_{t\wedge\tau}^{(1)}$.
This implies 
\begin{align*}
1-\E[Z_{T}^{(1)}] & =\P^{(1)}(\sigma^{(1)}\le T)\\
 & \ge\P^{(1)}(1+W^{(1)}\text{ hits zero before }T)-\P^{(1)}(\tau\le T)\\
 & =p(T)-\P^{(1)}(\tau\le T).
\end{align*}
It suffices to show that the right side is strictly positive. To this
end, we compute, for $t<\sigma^{(1)}$, 
\begin{align*}
L_{t} & =\exp\left(-(\beta-1)\int_{0}^{t}X_{s}\d W_{s}-\frac{1}{2}\int_{0}^{t}(\beta^{2}-1)X_{s}^{2}\d s\right)\\
 & =\exp\left(-(\beta-1)\int_{0}^{t}X_{s}\d W_{s}^{(1)}-\frac{1}{2}\int_{0}^{t}\left[(\beta^{2}-1)X_{s}^{2}-2(\beta-1)X_{s}^{2}\1{s<\tau}\right]\d s\right)\\
 & =\exp\left(-(\beta-1)\int_{0}^{t}X_{s}\d W_{s}^{(1)}-\frac{1}{2}\int_{0}^{t}(\beta-1)^{2}X_{s}^{2}\d s-\int_{t\wedge\tau}^{t}(\beta-1)X_{s}^{2}\d s\right)\\
 & =\Ecal\left(-(\beta-1)X\cdot W^{(1)}\right)_{t}\exp\left(-(\beta-1)\int_{t\wedge\tau}^{t}X_{s}^{2}\d s\right).
\end{align*}
Thus $L$ is a supermartingale under $\P^{(1)}$, which implies 
\[
\P^{(1)}(\tau<\infty)=\P^{(1)}\Big(\sup_{t\ge0}L_{t}\ge1+p(T)^{-1}\Big)\le(1+p(T)^{-1})^{-1}<p(T),
\]
by Doob's inequality. This proves $\E[Z_{T}^{(1)}]<1$.

We now establish $\E[Z_{T}^{(\beta)}]=1$. Similarly as before, let
$\P^{(\beta)}$ be the Föllmer measure associated with $Z^{(\beta)}$,
and let $W^{(\beta)}$ be a $\P^{(\beta)}$-Brownian motion such that
\[
W_{t}^{(\beta)}=W_{t}+\int_{0}^{t}\beta X_{s}\1{s<\tau}\d s,\qquad t<\sigma^{(\beta)}
\]
holds, where $\sigma^{(\beta)}$ is the explosion time of $Z^{(\beta)}$.
We then have 
\[
\frac{1}{Z_{t}^{(\beta)}}=\Ecal\left(\beta X\cdot W^{(\beta)}\right)_{t\wedge\tau},\qquad t<\sigma^{(\beta)},
\]
which implies $\P^{(\beta)}(\tau<\sigma^{(\beta)}<\infty)=0$ since
$Z^{(\beta)}$ does not move after~$\tau$. It also implies that
$\sigma^{(\beta)}$ can be written 
\[
\sigma^{(\beta)}=\lim_{n\to\infty}\inf\left\{ t\ge0:\int_{0}^{t\wedge\tau}X_{s}^{2}\d s\ge n\right\} .
\]
On the other hand, a calculation similar to the one in the first part
of the proof yields 
\[
\frac{1}{L_{t\wedge\tau}}=\Ecal\left((\beta-1)X\cdot W^{(\beta)}\right)_{t\wedge\tau},\qquad t<\sigma^{(\beta)},
\]
Hence on the event $\{\sigma^{(\beta)}<\tau\}$, we have $\lim_{t\uparrow\sigma^{(\beta)}}L_{t}=\infty$
(up to a $\P^{(\beta)}$-nullset). But this is impossible since, by
definition of $\tau$, $L_{t}$ cannot reach above $1+p(T)^{-1}$
prior to $\tau$. It follows that $\P^{(\beta)}(\sigma^{(\beta)}<\tau)=0$,
and consequently 
\[
\P^{(\beta)}(\sigma^{(\beta)}<\infty)=0.
\]
The proposition is proved. \end{proof}

\subsection{Example 1}

\label{S:ex1}

Using Proposition~\ref{P:counterex1} we can immediately construct
an example of a price process $S$ with nonnegative drift, which satisfies
both $\NFLVR_{C}$ and $\NUPBR$, but not $\NA$, on the time interval~$[0,T]$.
In other words, we have both an equivalent supermartingale measure
and a local martingale deflator, but no local martingale measure (and
certainly no martingale measure.) Indeed, in the notation of Proposition~\ref{P:counterex1},
set 
\[
S=\frac{1}{Z^{(1)}}.
\]
Then 
\[
\frac{\d S_{t}}{S_{t}}=X_{t}\1{t\le\tau}\d W_{t}+X_{t}^{2}\1{t\le\tau}\d t,
\]
so that $S$ is strictly positive with positive drift. The only candidate
density process is $Z^{(1)}$, which is a local martingale deflator
but fails to be a true martingale. Hence $\Dcal_{{\rm loc}}\ne\emptyset$
but $\Mcal_{{\rm loc}}=\emptyset$. On the other hand, $Z^{(\beta)}$
is a true martingale and thus induces a measure $\P^{(\beta)}$ that
is equivalent to $\P$ on $\Fcal_{T}$. Since 
\[
\frac{\d S_{t}}{S_{t}}=X_{t}\1{t\le\tau}\d W_{t}^{(\beta)}-(\beta-1)X_{t}^{2}\1{t\le\tau}\d t,
\]
it follows that $S$ is a supermartingale under $\P^{(\beta)}$, which
thus lies in $\Mcal_{{\rm sup}}$.

Note that $Z^{(1)}\in\Dcal_{{\rm loc}}$ has the property that $Z^{(1)}S=1$
is a true martingale. Hence by Theorem~\ref{T:reprchar} the price
process $S$ is supported by some constrained representative agent
equilibrium, and $\ND_{C}$ holds.

\subsection{Example 2}

\label{S:ex2}

In Theorem~\ref{T:reprchar}, the representative agent has stochastic
utility in general. We now modify the previous example slightly to
get a price process $S$ satisfying $\NFLVR_{C}$ and $\NUPBR$, but
not $\NA$, with the property that the constant strategy $\widehat{H}\equiv1$
is optimal for an agent with logarithmic utility $U(x)=\log x$. In
other words, in equilibrium the investor does not short the risky
asset. This log-utility agent is therefore a valid representative
agent for an equilibrium supporting the price process~$S$. The key
for this is to make the drift of $S$ strictly positive for all $t\in[0,T]$,
not just prior to~$\tau$.%
\footnote{Indeed, if $S$ is a supermartingale on $\rc\tau,T\rc$, then $\E[\log S_{T}]=\E[\log(S_{T}/S_{T\wedge\tau})]+\E[\log S_{T\wedge\tau}]\le\E[\log S_{T\wedge\tau}]$,
where Jensen's inequality and the assumption $\E[S_{T}\mid\Fcal_{T\wedge\tau}]\le S_{T\wedge\tau}$
were used. Hence it is always better to stop trading at $\tau$. Economically
this is because a risk-averse agent will not take on the risk inherent
in holding $S$, without being compensated by positive returns in
excess of the risk-free rate.%
} We define the candidate local martingale deflator by 
\[
Z=\Ecal\left(-(X\oo_{\lc0,\tau\rc}+\oo_{\rc\tau,T\rc})\cdot W\right),
\]
where $X$ solves $\d X_{t}=-X_{t}^{2}\d W_{t}$. The candidate density
process for an equivalent supermartingale measure is given by 
\[
Z^{{\rm sup}}=\Ecal\left(-(\beta X\oo_{\lc0,\tau\rc}+\oo_{\rc\tau,T\rc})\cdot W\right)
\]
for some $\beta>1$. Finally, the price process is defined by 
\[
S=\frac{1}{Z}.
\]

\begin{lemma} The process $Z$ is a strict local martingale, and
$Z^{{\rm sup}}$ is a true martingale. \end{lemma}

\begin{proof} We would like to apply Proposition~\ref{P:counterex1},
and for this we first change probability measure to $\widetilde{\P}$
given by $\d\widetilde{\P}=\widetilde{Z}_{T}\d\P$, where $\widetilde{Z}=\Ecal(\oo_{\rc\tau,T\rc}\cdot W)$
is a true martingale by Novikov's condition. The process $\d\widetilde{W}_{t}=\d W_{t}-\oo_{\rc\tau,T\rc}(t)\d t$
is then Brownian motion under $\widetilde{\P}$. Define $\widetilde{X}$
as the solution to $\d\widetilde{X}_{t}=-\widetilde{X}_{t}^{2}\d\widetilde{W}_{t}$,
let $\widetilde{L}_{t}=\Ecal(-\beta\widetilde{X}\cdot\widetilde{W})/\Ecal(-\widetilde{X}\cdot\widetilde{W})$,
and consider the stopping time $\widetilde{\tau}=\inf\{t\ge0:\widetilde{L}_{t}\ge1+p(T)^{-1}\}$.
Then, with 
\[
\widetilde{Z}^{(1)}=\Ecal\left(-\widetilde{X}\oo_{\lc0,\widetilde{\tau}\rc}\cdot\widetilde{W}\right)\qquad\text{and}\qquad\widetilde{Z}^{(\beta)}=\Ecal\left(-\beta\widetilde{X}\oo_{\lc0,\widetilde{\tau}\rc}\cdot\widetilde{W}\right),
\]
Proposition~\ref{P:counterex1} shows that $\widetilde{Z}^{(1)}$
is strict local martingale and $\widetilde{Z}^{(\beta)}$ is a true
martingale. This carries over to $\widetilde{Z}\widetilde{Z}^{(1)}$
and $\widetilde{Z}\widetilde{Z}^{(\beta)}$, respectively. Now, by
pathwise uniqueness of the defining SDE for $\widetilde{X}$ and $X$,
together with the fact that $\widetilde{W}_{t}=W_{t}$ for $t<\tau$,
it follows that $\widetilde{X}_{t}=X_{t}$ for $t<\tau$. Hence, for
$t<\tau$, we have $\widetilde{L}_{t}=L_{t}$ and thus $\widetilde{\tau}=\tau$.
Consequently, $\widetilde{Z}\widetilde{Z}^{(1)}=Z$ and $\widetilde{Z}\widetilde{Z}^{(\beta)}=Z^{{\rm sup}}$,
and this proves the lemma. \end{proof}

As in the previous example it is now straightforward to verify that
$S$ has strictly positive drift, $\Dcal_{{\rm loc}}=\{Z\}$, $Z^{{\rm sup}}\in\Dcal_{{\rm sup}}$,
and that $\Mcal_{{\rm loc}}=\emptyset$ since $Z$ is not a true martingale.
It thus remains to show that the constant strategy $\widehat{H}\equiv1$
is optimal under logarithmic utility. But this follows directly from
the definition of $S$, Jensen's inequality, and the fact that $Z$
is a local martingale deflator. Indeed, if $H$ is any $1$-admissible
strategy and setting $X_{T}=1+(H\cdot S)_{T}$, we have 
\[
\E[\log X_{T}-\log S_{T}]=\E[\log(Z_{T}X_{T})]\le\log(\E[Z_{T}X_{T}])\le0.
\]
Note that the nonnegativity of $H$ was not used. This means that
$\widehat{H}$ is optimal also among unconstrained, $1$-admissible
strategies. However, since $\NA$ fails, it is not optimal among {\em
all} admissible strategies. In this sense, the short sale constraint
is binding, because an unconstrained agent with access to any admissible
strategy could produce arbitrage and therefore (dramatically) improve
upon $\widehat{H}$. Any such strategy would involve short positions.
On the other hand, since the optimal strategy $\widehat{H}$ is strictly
positive, it is tempting to think that the short sale constraint is
not binding. As the above shows, this interpretation is misleading.

Another way of saying the same thing is the following: naively one
may conjecture that a strictly positive strategy could be improved
upon without violating the short sale constraint by adding ``a little
bit'' of an (unconstrained) arbitrage strategy. This would imply
that $\NA$ is incompatible with strictly positive optimal strategies.
Of course the above shows that this is not the case---but what is
interesting to note is that the analogous naive argument is correct
for $\NUPBR$; this (together with market clearing) is what drives
Theorem~\ref{T:NUPBR}.

This discussion suggests that optimal strategies need not change in
a continuous way with respect to changes in investment constraints.
Further investigation of this phenomenon is left for future research.

\section{Conclusion}

\label{S:concl}

In this paper we have considered constrained informational efficiency
of a financial market. That is, we study conditions under which a
given price process can be viewed as the outcome of some equilibrium
where agents are not allowed to short the risky asset. We find, somewhat
surprisingly, that constrained informational efficiency implies the
existence of a local martingale deflator. The explanation for this
result resides with the role of market clearing: at any point in time,
at least one agent is locally unconstrained. However, this intuition
does not carry over to arbitrage opportunities---indeed, we give an
example of an equilibrium with a single representative agent with
logarithmic utility, optimally holds the full supply of the risky
asset, but where arbitrage is nonetheless possible using unconstrained
strategies. Less surprisingly, $\NFLVR_{C}$ necessarily holds in
equilibrium, and $\ND_{C}$ holds whenever a representative agent
equilibrium exists. The existence of such an equilibrium is fully
characterized: we show that a constrained representative agent equilibrium
exists if and only if $\NFLVR_{C}$ holds and there is a local martingale
deflator turning the price process into a true martingale. Moreover,
in a multi-agent complete market equilibrium, these conditions are
satisfied whenever all agents have strictly positive holdings in the
risky asset.

The above shows that the inclusion of short sale constraints, arguably
a fairly modest modification of the fully frictionless framework,
yields a surprisingly rich set of phenomena. Nonetheless, much remains
to be done. In the future it would be interesting to investigate more
general trading constraints, multiple assets, and price processes
with jumps. We end by listing some open questions that naturally arose
from the analysis presented in the present paper. It would be enlightening
to see these issues resolved. 
\begin{enumerate}
\item If $H,K$ are two C-maximal strategies, is the sum $H+K$ again C-maximal?
Since each investor strategy $\widehat{H}^{k}$ is C-maximal in equilibrium
by Lemma~\ref{L:NFLVRC}, an affirmative answer to this question
would imply that $\widehat{H}^{1}+\cdots+\widehat{H}^{n}=1$ is again
C-maximal, so that $\ND_{C}$ always holds in equilibrium. 
\item Is there a constrained informationally efficient market that is not
consistent with any constrained representative agent equilibrium?
Equivalently (see Theorem~\ref{T:reprchar}), does there exist a
constrained equilibrium $(S;\ U_{k},x_{k},\widehat{H}^{k}:k=1,\ldots,n)$
such that $ZS$ is a strict local martingale for every $Z\in\Dcal_{{\rm loc}}$? 
\item If $ZS$ is a true martingale for some $Z\in\Dcal_{{\rm loc}}$, then
$\ND_{C}$ is satisfied. Does the converse hold? If not, is there
a different dual characterization of~$\ND_{C}$? 
\end{enumerate}
\bibliographystyle{abbrv}


\end{document}